\documentclass[11pt, a4paper]{article}

\usepackage{bm}
\usepackage{authblk}
\usepackage{amssymb, amsmath}
\usepackage{color}
\usepackage{tikz}

\usepackage[amsmath,thmmarks]{ntheorem}

\newtheorem{theorem}{Theorem}[section]
\newtheorem{proposition}[theorem]{Proposition}
\newtheorem{lemma}[theorem]{Lemma}
\newtheorem{corollary}[theorem]{Corollary}

\theorembodyfont{\rmfamily}
\newtheorem{remarks}{Remark}[section]
\newtheorem{example}{Example}[section]
\newtheorem{definition}{Definition}[section]

\theoremstyle{nonumberplain}
\theoremsymbol{$\square$}
\newtheorem{proof}{Proof}

\usepackage[ruled, vlined, linesnumbered]{algorithm2e}

\DeclareMathOperator{\level}{level}

\DeclareMathOperator{\supp}{supp}

\DeclareMathOperator{\algwang}{TriDecWang}

\DeclareMathOperator{\pop}{\sf pop}

\DeclareMathOperator{\ini}{ini}
\DeclareMathOperator{\lv}{lv}
\DeclareMathOperator{\red}{red}
\DeclareMathOperator{\tail}{tail}
\DeclareMathOperator{\prem}{prem}
\newcommand{\field}[1]{\mathbb{#1}}
\newcommand{\pset}[1]{\mathcal{#1}}
\newcommand{\point}[1]{\bm{#1}}
\newcommand{\p}[1]{\point{#1}}

\newcommand{\qnum}{\mathbb{Q}}
\newcommand{\kxring}{\field{K}[\point{x}]}

\newcommand{\kx}{\kxring}

\newcommand{\grobner}{Gr\"{o}bner }
\DeclareMathOperator{\zero}{\mathsf{Z}}
\newcommand{\fk}{\field{K}}

\begin{document}
\title{On the chordality of polynomial sets in triangular decomposition in top-down style\footnote{This work was partially supported by the National Natural Science Foundation of China (NSFC 11401018 and 11771034)}}

\author{Chenqi Mou and Yang Bai}
\affil{LMIB -- School of Mathematics and Systems Science / \authorcr
Beijing Advanced Innovation Center for Big Data and Brain Computing \authorcr
Beihang University, Beijing 100191, China\authorcr
\{chenqi.mou, yangbai\}@buaa.edu.cn}

\date{}

\maketitle

\begin{abstract}
In this paper the chordal graph structures of polynomial sets appearing in triangular decomposition in top-down style are studied when the input polynomial set to decompose has a chordal associated graph. In particular, we prove that the associated graph of one specific triangular set computed in any algorithm for triangular decomposition in top-down style is a subgraph of the chordal graph of the input polynomial set and that all the polynomial sets including all the computed triangular sets appearing in one specific simply-structured algorithm for triangular decomposition in top-down style (Wang's method) have associated graphs which are subgraphs of the the chordal graph of the input polynomial set. These subgraph structures in triangular decomposition in top-down style are multivariate generalization of existing results for Gaussian elimination and may lead to specialized efficient algorithms and refined complexity analyses for triangular decomposition of chordal polynomial sets. 
\end{abstract}


\section{Introduction}
\label{sec:intro}

This paper is inspired by the pioneering work on the connections between chordal graphs and triangular sets \cite{C2017c}, where the authors introduced the concept of chordal networks and proposed an algorithm for constructing chordal networks for polynomial sets based on the computation of triangular decomposition. In particular, the authors found that for input polynomial sets with chordal associated graphs, the elimination methods due to Wang become more efficient. It worths mentioning that the authors also studied the connections between chordal graphs and \grobner bases \cite{C2016e}, and found that the chordal structures of polynomial sets are destroyed in the computation of \grobner bases. 

The chordal graph structures have been studied and applied in the prediction of structures of matrices appearing in solving (especially sparse) linear systems. It is shown that the sparsity of the matrices handled by Gaussian elimination can be controlled if the associated graph of this input matrix is chordal \cite{P1961t,R1970t,G1994p}. %In fact, the applications of graph algorithms in solving linear algebra is quite a success \cite{KG2011g}.

Like the \grobner basis which has been greatly developed in its theory, methods, implementations, and applications \cite{B1965A,F1999A,f03a,CLO1998U}, the triangular set is also a powerful algebraic tool in the study on and computation of polynomials symbolically, especially for elimination theory and polynomial system solving \cite{w86z,GC1992s,k93g,w93e,a99t,W2001E,LI2010D,CM2012a}, with diverse applications \cite{Wu94m,c08c}. The process of decomposing a polynomial set into finitely many triangular sets or systems (probably with additional properties like being regular or normal, etc.) with associated zero and ideal relationships is called triangular decomposition of the input polynomial set. Triangular decomposition of polynomial sets can be regarded as multivariate generalization of Gaussian elimination for solving linear equations. 

The top-down strategy in triangular decomposition means that the variables appearing in the input polynomial set are handled in a strictly deceasing order, and it is a widely-used strategy in the design and implementations of algorithms for triangular decomposition. In particular, most algorithms for triangular decomposition due to Wang are in the top-down style \cite{w93e,w98d,w00c}. A Boolean algorithm for triangular decomposition in top-down style with refinement in the Boolean settings has also been proposed \cite{g12c}. The fact that elimination in it is performed in a strictly decreasing order makes triangular decomposition in top-down style the closest among all kinds of triangular decomposition to Gaussion elimination, in which the elimination of entries in different columns of the matrix is also performed in a strict order. 

The study in this paper arises naturally after summing up what are stated above: we study the chordal structures of polynomial sets appearing in the algorithms for triangular decomposition in top-down style, in particular the graph structures of the triangular sets computed by such algorithms. This is multivariate generalization of the study on the roles chordal structures play in Gaussian elimination, and it is highly non-trivial because in this polynomial (and thus nonlinear) case splitting occurs in the triangular decomposition which results in a complicated decomposition process. 

With the introduction of associated graphs of polynomial sets in Section~\ref{sec:pre}, we define a polynomial set to be chordal if its associated graph is chordal. A chordal graph implies a perfect elimination ordering of the vertexes, and we assume that the polynomial set $\pset{F}$ we want to decompose is chordal with the variables ordered as one perfect elimination ordering of its chordal associated graph $G(\pset{F})$. Under such assumptions, in this paper the graph structures of polynomial sets after reduction of one variable and all the variables in triangular decomposition in top-down style are exploited in Section~\ref{sec:perserve}. In particular, it is proved that the associated graph of one specific triangular set computed in an arbitrary algorithm for triangular decomposition in top-down style will be a subgraph of $G(\pset{F})$. Then in Section~\ref{sec:wang} for a specific simply-structured algorithm for triangular decomposition in top-down style, namely Wang's method, after reformulation of the underlying structures of its decomposition tree, we prove that each polynomial set occurring in the decomposition process has an associated graph being subgraph of $G(\pset{F})$, which directly implies that all the triangular sets computed by Wang's method have associated graphs which are subgraphs of $G(\pset{F})$. This paper ends with brief discussions on the applications of graph structures of polynomial sets on expressing the variable sparsity of polynomial sets and potential refined complexity analyses on triangular decomposition in top-down style in Section~\ref{sec:discussion}.

\section{Preliminaries}
\label{sec:pre}

Let $\fk$ be a field, and $\fk[x_1, \ldots, x_n]$ be the multivariate polynomial ring over $\fk$ in the variables $x_1, \ldots, x_n$. For the sake of simplicity, we write $(x_1, \ldots, x_n)$ as $\p{x}$ and $\fk[x_1, \ldots, x_n]$ as $\kx$. %Let $\p{i} = [i_1, \ldots, i_n]$ be a permutation of $[1, 2, \ldots, n]$. We use $\ord(\p{i})$ to denote the ordering of the variables $x_1, \ldots, x_n$ as $x_{i_1} < x_{i_2} < \cdots < x_{i_n}$, and $\ord_k(\p{i}) := \ord([i_1, \ldots, i_k])$ for an integer $k~(1 \leq k \leq n)$.

\subsection{Polynomial sets and associated graphs}
\label{sec:graph}

For a polynomial $F\in \kx$, define the (variable) \emph{support} of $F$, denoted by $\supp(F)$, to be the set of variables in $x_1, \ldots, x_n$ which effectively appear in $F$. For a polynomial set $\pset{F}\subset \kx$, $\supp(\pset{F}) := \cup_{F\in \pset{F}}\supp(F)$. 

The \emph{associated graph} $G(\pset{F})$ of a polynomial set $\pset{F} \subset \kx$ is an undirected graph constructed in the following way: 
\begin{itemize}
\item[(a)] The vertexes of $G(\pset{F})$ are the variables in $\supp(\pset{F})$. 
\item[(b)] There exists an edge connecting two vertexes $x_i$ and $x_j$ in $G(\pset{F})$ for $1\leq i\neq j \leq n$ if there exists one polynomial $F\in \pset{F}$ such that $x_i, x_j \in \supp(F)$. 
\end{itemize}

\begin{example}\label{ex:associated}
  The associated graphs of 
  \begin{equation*}
    \begin{split}
      \pset{P} &= \{x_2+x_1, x_3+x_1, x_4^2+x_2, x_4^3+x_3, x_5+x_2, x_5+x_3+x_2 \}\\
\pset{Q} &= \{x_2+x_1, x_3+x_1, x_3, x_4^2+x_2, x_4^3+x_3, x_5+x_2\}
    \end{split}
  \end{equation*}
are shown in Figure~\ref{fig:associated}.
    \begin{figure}[ht]
      \centering
\includegraphics[width=3cm,keepaspectratio]{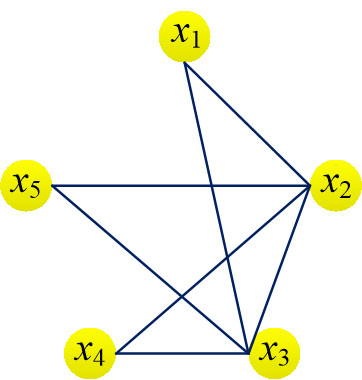}~\qquad
\includegraphics[width=3cm,keepaspectratio]{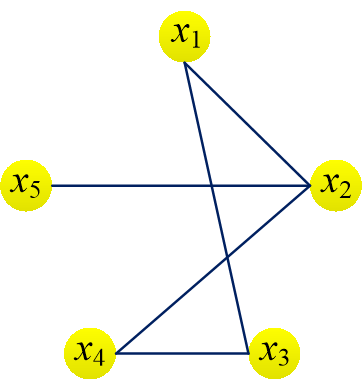}      
      \caption{The associated graphs $G(\pset{P})$ (left) and $G(\pset{Q})$ (right) in Example~\ref{ex:associated}}
      \label{fig:associated}
    \end{figure}
\end{example}

% \begin{proposition}\label{prop:divide}
%   Let $\pset{P} = \pset{P}_1 \cup \pset{P}_2 \subset \kx $ and $\pset{Q} = \pset{Q}_1 \cup \pset{Q}_2 \subset \kx$ with $\pset{P}_1 \cap \pset{P}_2 = \emptyset$ and $\pset{Q}_1 \cap \pset{Q}_2 = \emptyset$. If $G(\pset{P}_1)\subset G(\pset{Q}_1)$ and $G(\pset{P}_2) \subset G(\pset{Q}_2)$, then $G(\pset{P}) \subset G(\pset{Q})$.
% \end{proposition}

% \begin{proof}
% For any vertex $x_i$ in $G(\pset{P})$, there exists $P\in \pset{P}$ such that $x_i \in \supp(P)$. Without loss of generality, we can assume that $P\in \pset{P}_1$. Then $x_i \in G(\pset{P}_1) \subset G(\pset{Q}_1)$, and thus $x_i \in G(\pset{Q})$. 

% For any edge $x_ix_j$ in $G(\pset{P})$, there exists $P \in \pset{P}$ such that $x_i, x_j \in \supp(P)$. Under the assumption that $P\in \pset{P}_1$, we know that $x_ix_j$ is in $G(\pset{P}_1)$, and thus in $G(\pset{Q}_1)$ and in $G(\pset{Q})$.
% \end{proof}

\begin{definition}\label{def:poe}
Let $G = (V, E)$ be a graph with $V = \{x_1, \ldots, x_n\}$. Then an ordering $x_{i_1} < x_{i_2} < \cdots < x_{i_n}$ of the vertexes is called a \emph{perfect elimination ordering} of $G$ if for each $j=i_1, \ldots, i_n$, the restriction of $G$ on the following set
\begin{equation}\label{eq:smaller}
  X_j = \{x_j\} \cup \{x_k: x_k < x_j \mbox{ and } (x_k, x_j) \in E\}
\end{equation}
is a clique. A graph $G$ is said to be \emph{chordal} if there exists a perfect elimination ordering of it. 
\end{definition}

An equivalent condition for a graph $G = (V, E)$ to be \emph{chordal} is the following: for any cycle $C$ contained in $G$ of four or more vertexes, there is an edge $e\in E \setminus C$ connecting two vertexes in $C$. A chordal graph is also called a triangulated one. 

\begin{definition}
  A polynomial set $\pset{F} \subset \kx$ is said to be \emph{chordal} if its associated graph $G(\pset{F})$ is chordal. 
\end{definition}

\begin{example}
In Example~\ref{ex:associated} and Figure~\ref{fig:associated}, the associated graph $G(\pset{P})$ is chordal by definition and thus $\pset{P}$ is chordal, while $G(\pset{Q})$ is not. 
\end{example}

\subsection{Triangular sets and triangular decomposition}
\label{sec:ts}
Throughout this subsection the variables are ordered as $x_1 < \cdots < x_n$. For an arbitrary polynomial $F\in \kx$, denote the greatest variable appearing in $F$ by $\lv(F)$. Let $\lv(F) = x_k$. Write $F = Ix_k^d + R$ with $I\in \fk[x_1, \ldots, x_{k-1}]$, $R\in \fk[x_1, \ldots, x_k]$, and $\deg(R, x_k) < d$. Then the polynomial $I$ and $R$ are called the \emph{initial} and \emph{tail} of $F$ and denoted by $\ini(F)$ and $\tail(F)$ respectively. For two polynomial sets $\pset{F}, \pset{G}\subset \kx$, the set of common zeros of $\pset{F}$ is denoted by $\zero(\pset{F})$, and $\zero(\pset{F} / \pset{G}) := \zero(\pset{F}) \setminus \zero(\prod_{G\in \pset{G}}G)$. 

\begin{definition}\label{def:Tset-system}
  An ordered set of non-constant polynomials $\pset{T} = [T_1, \ldots, T_r] \subset \kx$ is called a \emph{triangular set} if $\lv(T_1) < \cdots < \lv(T_r)$. A tuple $(\pset{T}, \pset{U})$ with $\pset{T}, \pset{U} \subset \kx$ is called a \emph{triangular system} if $\pset{T}$ is a triangular set. 
\end{definition}

\begin{definition}\label{def:DintoSet}
 Let $\pset{F}\subset \kx$ be a polynomial set. Then a finite number of triangular sets $\pset{T}_1, \ldots, \pset{T}_r \subset \kx$ are called a \emph{decomposition} of $\pset{F}$ \emph{into triangular sets} if the zero relationship $\zero(\pset{F}) = \cup_{i=1}^r \zero(\pset{T}_i / \ini(\pset{T}_i))$ holds, where $\ini(\pset{T}_i) := \{\ini(T): T \in \pset{T}_i\}$.
\end{definition}

\begin{definition}\label{def:DintoSys}
  Let $\pset{F}\subset \kx$ be a polynomial set. Then a finite number of triangular systems $(\pset{T}_1,\pset{U}_1), \ldots, (\pset{T}_r, \pset{U}_r)$ are called a \emph{decomposition} of $\pset{F}$ \emph{into triangular systems} if the zero relationship $\zero(\pset{F}) = \cup_{i=1}^r \zero(\pset{T}_i / \pset{U}_i)$ holds. 
\end{definition}

As shown in Definitions~\ref{def:Tset-system}, \ref{def:DintoSet}, and \ref{def:DintoSys}, triangular systems are generalization of triangular sets. For a triangular system $(\pset{T}, \pset{U})$, $\pset{T}$ is a triangular set which represents the equations $\pset{T}=0$ while $\pset{U}$ is a polynomial set which represents the inequations $\pset{U}\neq 0$. 

In general, the process of computing a decomposition of a polynomial set into triangular sets or triangular systems is called \emph{triangular decomposition}. There exist many algorithms for decomposing polynomial sets into triangular sets or systems with different properties. One of the main strategies for designing such algorithms for triangular decomposition is to carry out reduction on polynomials containing the largest (unprocessed) variable until there is only one such polynomial with producing new polynomials whose leading variables are strictly smaller than the currently processed variable. 

% For any polynomial set $\pset{P} \subset \kx$ and a variable ordering $x_1 < \cdots < x_n$, we denote $\pset{P}^{(i)} = \{P\in \pset{P}:\, \lv(P) = x_i\}$. For any set $\Phi$ of polynomial sets, suppose that $\pset{F} \subset \kx$ is a polynomial set such that $|\pset{F}^{(i)}| = 1$ for $i=k+1, \ldots, n$ for some integer $k< n$. Then an algorithm for triangular decomposition is said to be in \emph{top-down style} if 

For any polynomial set $\pset{P} \subset \kx$ and an integer $i~(1\leq i \leq n)$, we denote $\pset{P}^{(i)} := \{P\in \pset{P}:\, \lv(P) = x_i\}$. The smallest integer $i~(1\leq i \leq n)$ such that $\#(\pset{P}^{(j)}) = 0$ or $1$ for each $j=i+1, \ldots, n$ is called the \emph{level} of $\pset{P}$ and denoted by $\level(\pset{P})$. For two polynomial sets $\pset{P}$ and $\pset{Q}$ in $\kx$, $\pset{P}$ is said to be of \emph{lower rank} than $\pset{Q}$ if either $\level(\pset{P}) < \level(\pset{Q})$ or $\level(\pset{P}) = \level(\pset{Q})$ but the minimal degree in $x_k$ of polynomials in $\pset{P}^{(k)}$ is strictly smaller than that in $\pset{Q}^{(k)}$, where $k=\level(\pset{P}) = \level(\pset{Q})$.

Let $\pset{F}$ be a polynomial set in $\kx$ and $\Phi$ be a set of polynomial sets, initialized with $\{\pset{F}\}$. Then an algorithm $\pset{A}$ for computing triangular decomposition of $\pset{F}$ is said to be in \emph{top-down style} if for each polynomial set $\pset{P} \in \Phi$ with $\level(\pset{P}) = k$, the algorithm $\pset{A}$ computes one polynomial set $\pset{P}'$ such that $\#(\pset{P}'^{(k)})=1$ and $\pset{P}'^{(i)} = \pset{P}^{(i)}$ for $i=k+1, \ldots, n$ and finitely many $\pset{Q}_1, \ldots, \pset{Q}_s$ which are all of lower ranks than $\pset{P}$ and are put into $\Phi$ for later computation with $\pset{A}$ such that $\zero(\pset{P}) = \zero(\pset{P}') \cup \zero(\pset{Q}_1) \cup \cdots \cup \zero(\pset{Q}_s)$.

\begin{remarks}
As mentioned in the introduction, algorithms for triangular decomposition in top-down style are polynomial generalization of Gaussian elimination for transforming a non-singular matrix into echelon form. The requirements on $Q_1, \ldots, Q_s$ above to have lower ranks than $\pset{P}$ guarantee termination of the algorithm $\pset{A}$.
\end{remarks}
% \section{Triangular sets induce chordal graphs}
% \label{sec:triangular}

% \redmark{The results in this section seem to be wrong! We perhaps need to extend the notion of graph to a weighted one.}

% \redmark{This section is not very related to the topic of this paper. Perhaps an individual paper should be written, together with keeping chordal structures with triangular decomposition.}

% Next we show that for a triangular set $\pset{T} = [T_1, \ldots, T_r] \subset \kx$ with respect to the variable ordering $\ord([1, 2, \ldots, n])$, the ordering $\ord([1, 2, \ldots, n])$ is a perfect elimination ordering for the associated graph $G(\pset{T})$ of $\pset{T}$, and thus $G(\pset{T})$ is chordal. 

% First we note from the definition of an associated graph $G(\pset{F})$ of a polynomial set $\pset{F}$ that for each polynomial $F\in \pset{F}$, the restriction of $G(\pset{F})$ on the set of $\supp(F)$ is naturally a clique, where $\supp(F)$ denotes the support of $F$ (the variables in $x_1, \ldots, x_n$ which effectively appear in $F$). 

% Let $\pset{T} = [T_1, \ldots, T_r] \subset \kx$ with respect to the variable ordering $\ord([1, 2, \ldots, n])$. For each $j = 1, \ldots, n$, one can easily check for $G(\pset{T})$ that the set in \eqref{eq:smaller} in Definition~\ref{def:peo} is either 

\section{Chordality of polynomial sets in general triangular decomposition in top-down style}
\label{sec:perserve}

In this section, the graph structures of polynomial sets in an arbitrary algorithm for triangular decomposition in top-down style are studied when the input polynomial set is chordal. We start this section with the connections between the associated graphs of a triangular set reduced from a chordal polynomial set and the chordal associated graph of the input polynomial set. 

\begin{proposition}\label{prop:woReduction}
Let $\pset{P} \subset \kx$ be a chordal polynomial set with $x_1 < \cdots < x_n$ as one perfect elimination ordering, and for $i=1, \ldots, n$, let $T_i\in \kx$ be a polynomial such that $\lv(T_i) = x_i$ and $\supp(T_i) \subset \supp(\pset{P}^{(i)})$ ($T_i$ is set null if $\pset{P}^{(i)} = \emptyset$). Then $\pset{T} = [T_1, \ldots, T_n]$ is a triangular set, and $G(\pset{T}) \subset G(\pset{P})$. In particular, if $\supp(T_i) = \supp(\pset{P}^{(i)})$ for $i=1, \ldots, n$, then $G(\pset{T}) = G(\pset{P})$.
\end{proposition}

\begin{proof}
It is straightforward that $\pset{T}$ is a triangular set with $\lv(T_i) = x_i$ if $\pset{P}^{(i)}\neq \emptyset$ for $i=1, \ldots, n$. 

%($G(\pset{T}) \subseteq G(\pset{P})$) but $\{x_i, x_j\} \not \subseteq \supp(\pset{T}_{k+l})$ for $l=1, \ldots, n-k$
For any edge $(x_i,x_j) \in G(\pset{T})$, there exists an integer $k~(i, j \leq k \leq n)$ such that $x_i, x_j \in \supp(T_k)$. Then $x_i, x_j \in \supp(\pset{P}^{(k)})$, and thus $(x_i, x_k) \in G(\pset{P})$ and $(x_j, x_k) \in G(\pset{P})$. Since $G(\pset{P})$ is chordal with $x_1 < \ldots < x_n$ as a perfect elimination ordering and $x_i < x_k$, $x_j < x_k$, we know that $(x_i, x_j) \in G(\pset{P})$ by Definition~\ref{def:poe}. This proves the inclusion $G(\pset{T}) \subset G(\pset{P})$.

In particular, when $\supp(T_i) = \supp(\pset{P}^{(i)})$ for $i=1, \ldots, n$, next we show the inclusion $G(\pset{T}) \supset G(\pset{P})$, which implies the equality $G(\pset{T}) = G(\pset{P})$. For any $(x_i, x_j) \in G(\pset{P})$, there exists an integer $k$ and a polynomial $P$ such that $x_i, x_j \in \supp(P)$ with $P \in \pset{P}^{(k)}$. Since $\supp(P) \subseteq \supp(\pset{T}_k)$, $x_i, x_j \in \supp(T_k)$, and thus $(x_i, x_j) \in G(\pset{T})$. 
\end{proof}

\begin{example}\label{ex:ts}
  Proposition~\ref{prop:woReduction} does not necessarily hold in general if the polynomial set $\pset{P}$ is not chordal. Consider the same $\pset{Q}$ as in Example~\ref{ex:associated} whose associated graph $G(\pset{Q})$ is not chordal. Let 
$$\pset{T} = [x_2+x_1, x_3+x_1, -x_2x_4+x_3, x_5+x_2].$$
Then one can check that for $i = 2, \ldots, 5$, $\supp(\pset{T}^{(i)}) = \supp(\pset{Q}^{(i)})$, but the associated graph $G(\pset{T})$, as shown in Figure~\ref{fig:ts}, is not a subgraph of $G(\pset{Q})$. 
    \begin{figure}[ht]
      \centering
\includegraphics[width=3cm,keepaspectratio]{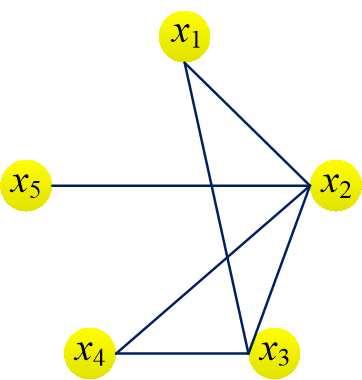}      
      \caption{The associated graph $G(\pset{T})$ in Example~\ref{ex:ts}}
\label{fig:ts}
    \end{figure}
\end{example}

Proposition~\ref{prop:woReduction} above relates the associated graph of a triangular set and that of a chordal polynomial set when their variables satisfy certain conditions. The following theorem is for relating the associated graph of a chordal polynomial set before and after one kind of commonly used reduction in triangular decomposition.

\begin{theorem}\label{thm:reduction}
Let $\pset{P} \subset \kx$ be a chordal polynomial set such that $\pset{P}^{(n)} \neq \emptyset$ and $x_1 < \ldots < x_n$ is one perfect elimination ordering. Let $T \in \kx$ be a polynomial such that $\lv(T) = x_n$ and $\supp(T) \subset \supp(\pset{P}^{(n)})$, and $\pset{R} \subset \kx$ be a polynomial set such that $\supp(\pset{R}) \subset \supp(\pset{P}^{(n)}) \setminus \{x_n\}$. Then for the polynomial set $\tilde{\pset{P}} = \{\tilde{\pset{P}}^{(1)} , \ldots, \tilde{\pset{P}}^{(n-1)}, T\}$, where $\tilde{\pset{P}}^{(k)} = \pset{P}^{(k)} \cup \pset{R}^{(k)}$ for $k=1, \ldots, n-1$, we have $G(\tilde{\pset{P}}) \subset G(\pset{P})$. In particular, if $\supp(T) = \supp(\pset{P}^{(n)})$, then $G(\tilde{\pset{P}}) = G(\pset{P})$.
\end{theorem}

\begin{proof}
To prove the inclusion $G(\tilde{\pset{P}}) \subset G(\pset{P})$, it suffices to show that for each edge $(x_i, x_j) \in G(\tilde{\pset{P}})$, we have $(x_i, x_j) \in G(\pset{P})$. For an arbitrary edge $(x_i, x_j) \in G(\tilde{\pset{P}})$, there exists a polynomial $P\in \tilde{\pset{P}}$ and an integer $k~(i, j \leq k \leq n)$ such that $x_i, x_j \in \supp(P)$ and $P\in \tilde{\pset{P}}^{(k)}$. 

If $k=n$, then $x_i, x_j \in \supp(T)$, and by $\supp(T) \subset \supp(\pset{P}^{(n)})$ we have $x_i, x_j \in \supp(\pset{P}^{(n)})$. This implies that $(x_i, x_n), (x_j, x_n) \in G(\pset{P}^{(n)}) \subset G(\pset{P})$ and by the chordality of $G(\pset{P})$ we have $(x_i, x_j) \in G(\pset{P})$.

Else if $k<n$, then by $\tilde{\pset{P}}^{(k)} = \pset{P}^{(k)} \cup \pset{R}^{(k)}$ there are two cases for $P$ accordingly. When $P\in \pset{P}^{(k)} \subset \pset{P}$, clearly $(x_i, x_j) \in G(\pset{P})$; when $P\in \pset{R}^{(k)}$, we have $x_i, x_j \in \supp(\pset{R}^{(k)}) \subset \supp(\pset{P}^{(n)})$, and thus $(x_i, x_n), (x_j, x_n) \in G(\pset{P}^{(n)}) \subset G(\pset{P})$, and the chordality $G(\pset{P})$ implies $(x_i, x_j) \in G(\pset{P})$. 

In particular, if $\supp(T) = \supp(\pset{P}^{(n)})$, then by $G(\pset{P}^{(k)}) \subset G(\tilde{\pset{P}}^{(k)})$ for $k=1, \ldots, n-1$ and $G(\pset{P}^{(n)})\subset G(T)$, we have $G(\pset{P}) \subset G(\tilde{\pset{P}})$. This proves the equality $G(\tilde{\pset{P}}) = G(\pset{P})$.
\end{proof}

\begin{example}\label{ex:oneReduction}
  Let $\pset{P}$ be the chordal polynomial set as in Example~\ref{ex:associated}. Then $\pset{P}^{(5)} = \{x_5+x_2, x_5+x_3+x_2\}$. If we take $T = x_5+x_2$, and $\pset{R} = \{\prem(x_5+x_3+x_2, x_5+x_2)\} = \{x_3\}$, then $\tilde{\pset{P}}$ equals $\pset{Q}$ in Example~\ref{ex:associated}, and $G(\tilde{\pset{P}})$ is a (strict) subgraph of $G(\pset{P})$; If we take $T = x_5+x_3+x_2$, and $\pset{R} = \{\prem(x_5+x_2, x_5+x_3+x_2)\} = \{-x_3\}$, then $\supp(T) = \supp(\pset{P}^{(5)})$ and thus $G(\tilde{\pset{P}}) = G(\pset{P})$.
% \begin{figure}[ht]
%       \centering
% \includegraphics[width=4.2cm,keepaspectratio]{pic/3-4.png}~
% \includegraphics[width=4.5cm,keepaspectratio]{pic/3-2.png}      
%       \caption{The associated graphs $G(\pset{P})$ (left) and $G(\pset{Q})$ (right)}
%       \label{fig:associated}
%     \end{figure}
\end{example}

% \redmark{Before a main theorem which is essentially about the preservation of chordality in the triangular decomposition in a top-down style, we formalize the reduction process as follows. }

Next we introduce some notations to formulate the reduction process in Theorem~\ref{thm:reduction}. Denote the power set of a set $S$ by $2^S$. For an integer $i~(1\leq i \leq n)$, let $f_i$ be a mapping
 \begin{equation}\label{eq:mapping}
   \begin{split}
 f_i: 2^{\fk[\p{x}_i] \setminus \fk[\p{x}_{i-1}]} &\rightarrow (\fk[\p{x}_i] \setminus \fk[\p{x}_{i-1}]) \times 2^{\fk[\p{x}_{i-1}]}\\
 \pset{P} &\mapsto (T, \pset{R})     
   \end{split}
 \end{equation}
such that $\supp(T) \subset \supp(\pset{P}) $ and $\supp(\pset{R}) \subset \supp(\pset{P})$, where $\fk[\p{x}_0]$ is understood as $\fk$. For a polynomial set $\pset{P} \subset \kx$ and a fixed integer $i~(1\leq i \leq n)$, suppose that $(T_i, \pset{R}_i) = f_i(\pset{P}^{(i)})$ for some $f_i$ as stated above. For $j=1, \ldots, n$, define the polynomial set
\begin{equation*}
\red_i(\pset{P}^{(j)}) := \left\{
    \begin{tabular}[l]{ll}
      $\pset{P}^{(j)}$, & if $j > i$  \\
      $\{T_i\}$, & if $j=i$ \\
      $\pset{P}^{(j)} \cup \pset{R}_i^{(j)}$, & if $j<i$
    \end{tabular}
\right.
\end{equation*}
and $\red_i(\pset{P}) := \cup_{j=1}^n \red_i(\pset{P}^{(j)})$. In particular, write 
\begin{equation}
  \label{eq:redBar}
  \overline{\red}_i(\pset{P}) := \red_{i}(\red_{i+1}(\cdots (\red_n(\pset{P}))\cdots))
\end{equation}
for simplicity. 

Here $\red_i(\cdot)$ denotes the result of reduction with respect to $x_i$ and $\overline{\red}_i(\cdot)$ denotes the result of successive reduction with respect to $x_n, x_{n-1}, \ldots, x_i$. Following the above terminologies, Theorem~\ref{thm:reduction} can be reformulated as $G(\red_n(\pset{P})) \subset G(\pset{P})$, and the equality holds if $\supp(T_n) = \supp(\pset{P}^{(n)})$. 

%\redmark{how to deal with the polynomial set of degree 1!? We do not care.}

\begin{proposition}\label{prop:left}
Let $\pset{P} \subset \kx$ be a chordal polynomial set with $x_1 < \cdots < x_n$ as one perfect elimination ordering. For each $i~(1\leq i \leq n)$, suppose that $(T_i, \pset{R}_i) = f_i(\overline{\red}_{i+1}(\pset{P})^{(i)})$ for some $f_i$ as in \eqref{eq:mapping} and $\supp(T_i) = \supp(\overline{\red}_{i+1}(\pset{P})^{(i)})$, where $\overline{\red}_{n+1}(\pset{P})$ is understood as $\pset{P}$. Then $G(\overline{\red}_{1}(\pset{P})) = G(\pset{P})$. 
\end{proposition}
%Define $\pset{P}^{\{n-1\}} = \pset{P}'$ as in Lemma~\ref{lem:reduction} with some $\pset{R}\subset \kx$; for $k=n-2, \ldots, 1$, define inductively $\pset{P}^{\{k\}} := (\bar{\pset{P}}^{\{k+1\}})' \cup \tilde{\pset{P}}^{\{k+1\}}$,
% where 
% \begin{equation*}
%   \begin{split}
%     \bar{\pset{P}}^{\{k+1\}} &= \pset{P}^{\{k+1\}} \cap \fk[x_1, \ldots, x_{k+1}],\\
%     \tilde{\pset{P}}^{\{k+1\}} &= \pset{P}^{\{k+1\}}\setminus \bar{\pset{P}}^{\{k+1\}},
%   \end{split}
% \end{equation*}
% and $(\bar{\pset{P}}^{\{k+1\}})'$ is computed as in Lemma~\ref{lem:reduction} with some $\pset{R}_{k+1}$ such that $\supp(\pset{R}_{k+1}) \subset \supp(\pset{P}^{\{k+1\}(k+1)}) \setminus \{x_{k+1}\}$.
% Then $G(\pset{P}^{\{1\}}) = G(\pset{P})$.

\begin{proof}
Repeated use of Theorem~\ref{thm:reduction} implies
$$G(\pset{P}) = G(\red_n(\pset{P})) = G(\overline{\red}_{n-1}(\pset{P})) = \cdots = G(\overline{\red}_1(\pset{P})),$$
and the conclusion follows. 
% $G(\red{}) = G(\bar{\pset{P}}^{\{k+1\}})$.
% Since $\pset{P}^{\{k\}}$ and $\pset{P}^{\{k+1\}}$ share the same part $\tilde{\pset{P}}^{\{k+1\}}$, $G(\pset{P}^{\{k\}}) = G(\pset{P}^{\{k+1\}})$. By induction the conclusion $G(\pset{P}^{\{1\}}) = G(\pset{P})$ follows. 
\end{proof}

\begin{remarks}
  Note that $\overline{\red}_1(\pset{P})$ forms a triangular set after reordering if $\overline{\red}_1(\pset{P})$ does not contain any non-zero constant. Indeed, the reduction process to compute this triangular set is commonly used in algorithms for triangular decomposition in top-down style, and the mapping $f_i$ in \eqref{eq:mapping} is abstraction of specific reductions used in different kinds of algorithms for triangular decomposition \cite{JLW2013a}. For example, one specific kind of such reduction is performed by using pseudo-divisions, and in this case $\pset{R}$ in \eqref{eq:mapping} consists of pseudo-remainders which do not contain $x_i$. 
\end{remarks}

  Proposition~\ref{prop:left} holds because after every reduction $G(\overline{\red}_i(\pset{P}))$ remains the same as the chordal graph $G(\pset{P})$, and thus the hypotheses of Theorem~\ref{thm:reduction} remain satisfied. If we loosen the condition $\supp(T_i) = \supp(\overline{\red}_{i+1}(\pset{P})^{(i)})$ in Proposition~\ref{prop:left} to $\supp(T_i) \subset \supp(\overline{\red}_{i+1}(\pset{P})^{(i)})$, then in general we will not have
$$G(\overline{\red}_1(\pset{P})) \subset \cdots \subset G(\overline{\red}_{n-1}(\pset{P})) \subset G(\red_n(\pset{P})) \subset G(\pset{P}),$$
as shown by the following example (though the last inclusion always holds because $G(\pset{P})$ is chordal). 

\begin{example}\label{ex:successiveInc}
Let us continue with Example~\ref{ex:oneReduction} with $\pset{P}$ and $\pset{Q} = \red_5(\pset{P})$, where $G(\pset{Q}) \subsetneq G(\pset{P})$. Take 
\begin{equation*}
  \begin{split}
 T_4 &= \prem(x_4^3+x_3, x_4^2+x_2) = -x_2x_4+x_3, \\
\pset{R}_4 &= \{\prem(x_4^2+x_2, -x_2x_4+x_3)\} = \{x_3^2-x_2^3\},
  \end{split}
\end{equation*}
then 
$$\pset{Q}' := \overline{\red}_{4}(\pset{P}) = \{x_2+x_1, x_3+x_1, x_3^2-x_2^3, x_3, -x_2x_4+x_3, x_5+x_2\}.$$
The associated graph $G(\pset{Q}')$ is shown below. Note that $G(\pset{Q}') \not \subset G(\pset{Q})$ but $G(\pset{Q}') \subset G(\pset{P})$. 
    \begin{figure}[ht]
      \centering
\includegraphics[width=3cm,keepaspectratio]{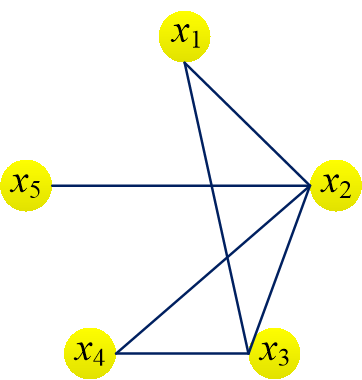} 
      \caption{The associated graph $G(\pset{Q}')$ in Example~\ref{ex:successiveInc}}
      \label{fig:needCompletion}
    \end{figure}
\end{example}

Despite of this example where successive inclusions of the associated graphs in the reduction chain does not hold, it can be proved that for each $i=n, \ldots, 1$, $G(\overline{\red}_i(\pset{P}))$ is a subgraph of the original graph $G(\pset{P})$.

\begin{lemma}\label{prop:subgraph}
  Let $\pset{P} \subset \kx$ be a chordal polynomial set with $x_1 < \cdots < x_n$ as one perfect elimination ordering and $\overline{\red}_i(\pset{P})$ be defined in \eqref{eq:redBar} for $i=n, \ldots, 1$. Then for each $i=n, \ldots, 1$ and any two variables $x_p$ and $x_q$, if there exists an integer $k$ such that $x_p, x_q \in \supp(\overline{\red}_i(\pset{P})^{(k)})$, then $(x_p, x_q) \in G(\pset{P})$. 
\end{lemma} 

\begin{proof}
  We induce on the integer $i$. In the case $i=n$, from the proof of Theorem~\ref{thm:reduction} one can easily find  that the proposition is true. Now suppose that the proposition holds for $i=j~(< n)$, and next we prove that it also holds for $i=j-1$, namely for any $x_p$ and $x_q$, if there exists $k~(\geq p, q)$ such that $x_p, x_q \in \supp(\overline{\red}_{j-1}(\pset{P})^{(k)})$, then $(x_p, x_q) \in G(\pset{P})$. 

First by $\overline{\red}_{j-1}(\pset{P}) = \red_{j-1}(\overline{\red}_j(\pset{P}))$ we know that there exists a polynomial set $\tilde{\pset{R}}$ such that
$$\overline{\red}_{j-1}(\pset{P})^{(k)} = \overline{\red}_j(\pset{P})^{(k)} \cup \tilde{\pset{R}}^{(k)}$$
and $\supp(\tilde{\pset{R}}) \subset \supp(\overline{\red}_j(\pset{P})^{(j-1)}) \setminus \{x_{j-1}\}$. 

(a) If $\tilde{\pset{R}} = \emptyset$, then $x_p, x_q \in \supp(\overline{\red}_j(\pset{P})^{(k)})$, and by the induction assumption we know that $(x_p, x_q) \in G(\pset{P})$. 

(b) If $\tilde{\pset{R}} \neq \emptyset$, then $x_k \in \supp(\tilde{\pset{R}}^{(k)}) \subset \supp(\overline{\red}_j(\pset{P})^{(j-1)})$. Next we consider the following three cases. 

Case (1): $x_p, x_q \in \supp(\overline{\red}_j(\pset{P})^{(k)})$: with the same argument as in (a) we know that $(x_p, x_q) \in G(\pset{P})$. 

Case (2): $x_p, x_q \in \supp(\tilde{\pset{R}}^{(k)}) \subset \supp(\overline{\red}_j(\pset{P})^{(j-1)})$: by the induction assumption we know that $(x_p, x_q) \in G(\pset{P})$. 

Case (3): $x_p \in \supp(\overline{\red}_j(\pset{P})^{(k)})$ and $x_q \in \supp(\tilde{\pset{R}}^{(k)})$ $\subset \supp(\overline{\red}_j(\pset{P})^{(j-1)})$: Since $x_p, x_k \in \supp(\overline{\red}_j(\pset{P})^{(k)})$, by the induction assumption we have $(x_p, x_k) \in G(\pset{P})$; since $x_q, x_k \in \supp(\overline{\red}_j(\pset{P})^{(j-1)})$, by the induction assumption we have $(x_q, x_k) \in G(\pset{P})$. Then by the chordality of $\pset{P}$, $(x_p, x_k) \in G(\pset{P})$ and $(x_q, x_k) \in G(\pset{P})$ imply that $(x_p, x_q) \in G(\pset{P})$. 

This ends the proof of this proposition with induction on $i$. 
\end{proof}

\begin{theorem}\label{thm:subgraph}
  Let $\pset{P} \!\subset\! \kx$ be a chordal polynomial set with $x_1 < \cdots < x_n$ as one perfect elimination ordering and $\overline{\red}_i(\pset{P})$ be defined in \eqref{eq:redBar} for $i=n, \ldots, 1$. Then for each $i=n, \ldots, 1$, $G(\overline{\red}_i(\pset{P})) \subset G(\pset{P})$.
\end{theorem}

\begin{proof}
  By the construction of $\overline{\red}_i(\pset{P})$, we know that all the vertexes of $G(\overline{\red}_i(\pset{P}))$ are also vertexes of $G(\pset{P})$. For each edge $(x_p, x_q) \in G(\overline{\red}_i(\pset{P}))$, there exists an integer $k~(p, q \leq k \leq n)$ and a polynomial $P$ such that $x_p, x_q \in \supp(P)$ and $P \in \overline{\red}_i(\pset{P})^{(k)}$. Then by Lemma~\ref{prop:subgraph}, we know that $(x_p, x_q) \in G(\pset{P})$, and thus $G(\overline{\red}_i(\pset{P})) \subset G(\pset{P})$. 
\end{proof}

\begin{corollary}\label{cor:subgraph2}
  Let $\pset{P} \subset \kx$ be a chordal polynomial set with $x_1 < \cdots < x_n$ as one perfect elimination ordering and $\overline{\red}_i(\pset{P})$ be defined in \eqref{eq:redBar} for $i=n, \ldots, 1$. If $\pset{T} := \overline{\red}_1(\pset{P})$ does not contain any nonzero constant, then $\pset{T}$ forms a triangular set such that $G(\pset{T}) \subset G(\pset{P})$. 
\end{corollary}

\begin{remarks}
Corollary~\ref{cor:subgraph2} tells us that under the conditions that the input polynomial set is chordal and the variable ordering is one perfect elimination ordering, the associated graph of one specific triangular set computed in any algorithm for triangular decomposition in top-down style is a subgraph of the associated graph of the input polynomial set. In fact, this triangular set is the ``main branch'' in the triangular decomposition in the sense that other branches are obtained by adding additional constrains in the splitting in the process of triangular decomposition. 

Note that in the case when the input polynomial set $\pset{P}$ is not chordal, a process of chordal completion can be carried out on $G(\pset{P})$ to make it a chordal graph (in the worst case this chordal completion results in a complete graph which is trivially chordal). Then the conditions of Corollary~\ref{cor:subgraph2} will be satisfied after this chordal completion. 
\end{remarks}

The chordality of any triangular sets other than the specific one above in a triangular decomposition computed by an algorithm in top-down style is dependent on the strategy how splitting occurs in the algorithm. Therefore in the next section we focus on Wang's method, one specific algorithm for triangular decomposition in top-down style, and prove that the associated graphs of all the triangular sets computed by Wang's method are subgraphs of the associated graph of a chordal input polynomial set.

\section{Chordality of polynomial sets computed by Wang's method}
\label{sec:wang}

A simply-structured algorithm was proposed by Wang for triangular decomposition in top-down style in 1993 \cite{w93e}, which is referred to as Wang's method in the literatures (see. e.g., \cite{a99ta}). In this section the chordaility of polynomial sets in the decomposition process of Wang's method is studied. 

\subsection{Restatement of Wang's method}
\label{sec:reformulation}

For the self-containness of this paper, Wang's method for triangular decomposition is outlined in Algorithm~\ref{alg:wang} below. In this algorithm, the data structure $(\pset{P}, \pset{Q}, i)$ is used to represent two polynomial sets $\pset{P}$ and $\pset{Q}$ such that $\#(\pset{P}^{(j)}) =1$ or $0$ for $j=i+1, \ldots, n$, and the subroutine $\pop(\Phi)$ returns an element in $\Phi$ and then remove it from $\Phi$. 

\begin{algorithm}[!ht]
\caption{Wang's method for triangular decomposition~~~$\Psi:=\algwang(\pset{F})$}
\label{alg:wang} 

\KwIn{$\pset{F}$, a polynomial set in $\kx$}

\KwOut{$\Psi$, a set of finitely many triangular systems which form a triangular decomposition of $\pset{F}$}

\BlankLine
$\Phi := \{(\pset{F}, \emptyset, n)\}$\;

\While{$\Phi \neq \emptyset$}
{
   $(\pset{P}, \pset{Q}, i) := \pop(\Phi)$\;
   \If{$i=0$}
   {
      $\Psi := \Psi \cup \{(\pset{P}, \pset{Q})\}$\;
      Break\;
   }

   \While{$\#(\pset{P}^{(i)}) > 1$}
   {
      $T:=$ a polynomial in $\pset{P}^{(i)}$ with minimal degree in $x_i$\; \label{line:minimal}
      $\Phi := \Phi \cup \{(\pset{P} \setminus \{T\} \cup \{\ini(T), \tail(T)\}, \pset{Q}, i)\}$\;
      $\overline{\pset{P}} := \pset{P}^{(i)} \setminus \{T\}$\;
      $\pset{P} := \pset{P} \setminus \overline{\pset{P}}$\;
      \For{$P\in \pset{P}^{(i)}$}
      {
         $\pset{P} :=  \pset{P} \cup \{\prem(P, T)\}$\;
      }      
      $\pset{Q} := \pset{Q} \cup \{\ini(T)\}$\;
   }
   $\Phi := \Phi \cup \{(\pset{P}, \pset{Q}, i-1)\}$\;

}

\For{$(\pset{P}, \pset{Q}) \in \Psi$}
{
   \If{$\pset{P}$ contains a non-zero constant}
   {
      $\Psi := \Psi \setminus \{(\pset{P}, \pset{Q})\}$
   }
}

\Return $\Psi$
\end{algorithm}

As shown in Algorithm~\ref{alg:wang}, for each $(\pset{P}, \pset{Q}, i)$ picked from $\Phi$, if $\pset{P}$ is already a triangular set (namely $i=0$, which means $\#(\pset{P}^{(j)})$ is either 0 or 1 for $j=1, \ldots, n$ and $\pset{P}$ contains no non-zero constant), then $(\pset{P}, \pset{Q})$ is included in the output $\Psi$. Note that $\pset{Q}$ in Algorithm~\ref{alg:wang} is for collecting the inequations $\ini(T)\neq 0$ for $T$ in Line~\ref{line:minimal} of Algorithm~\ref{alg:wang} and it does not impose any influence on the graph structures of the triangular sets computed by Wang's method. 

The decomposition process in Wang's method (Algorithm~\ref{alg:wang}) applied to $\pset{F}$ can be viewed as a binary tree with its root as $(\pset{F}, \emptyset, n)$. The nodes of this binary tree are all the tuples $(\pset{P}, \pset{Q}, i)$ picked from $\Phi$, and each node $(\pset{P}, \pset{Q}, i)$ has two children $(\pset{P}', \pset{Q}', i)$ and $(\pset{P}'', \pset{Q}'', i)$ where 
\begin{equation*}
  \begin{split}
    \pset{P}' &:= \pset{P} \setminus \pset{P}^{(i)} \cup \{T\} \cup \{\prem(P, T): P \in \pset{P}\}, \quad     \pset{Q}' := \pset{Q} \cup \{\ini(T)\}, \\
    \pset{P}'' &:= \pset{P} \setminus \{T\} \cup \{\ini(T), \tail(T)\}, \qquad\quad\quad\quad\quad\!    \pset{Q}'' := \pset{Q},
  \end{split}
\end{equation*}
with $T$ as a polynomial in $\pset{P}^{(i)}$ with minimal degree in $x_i$. In fact, the left child node $(\pset{P}', \pset{Q}', i)$ corresponds to the case when $\ini(T)\neq 0$ and thus reduction of $\pset{P}^{(i)}$ are performed with respect to $T$; while the right child node $(\pset{P}'', \pset{Q}'', i)$ corresponds to the case $\ini(T) = 0$ and thus $T$ is replaced by $\ini(T)$ and $\tail(T)$ (since $T=0$ and $\ini(T)=0$ imply $\tail(T)=0$). %Note that in either case, the degree of the polynomial in $\pset{P}'^{(i)}$ or $\pset{P}''^{(i)}$ with minimal degree in $x_i$ is strictly smaller than that of $T$ and thus the binary tree has a finite depth. 

The binary decomposition tree for Wang's method and the splitting at one node are illustrated in Figures~\ref{fig:decTree} and \ref{fig:oneNode} respectively. 

    \begin{figure}[ht]
      \centering
\includegraphics[width=9cm,keepaspectratio]{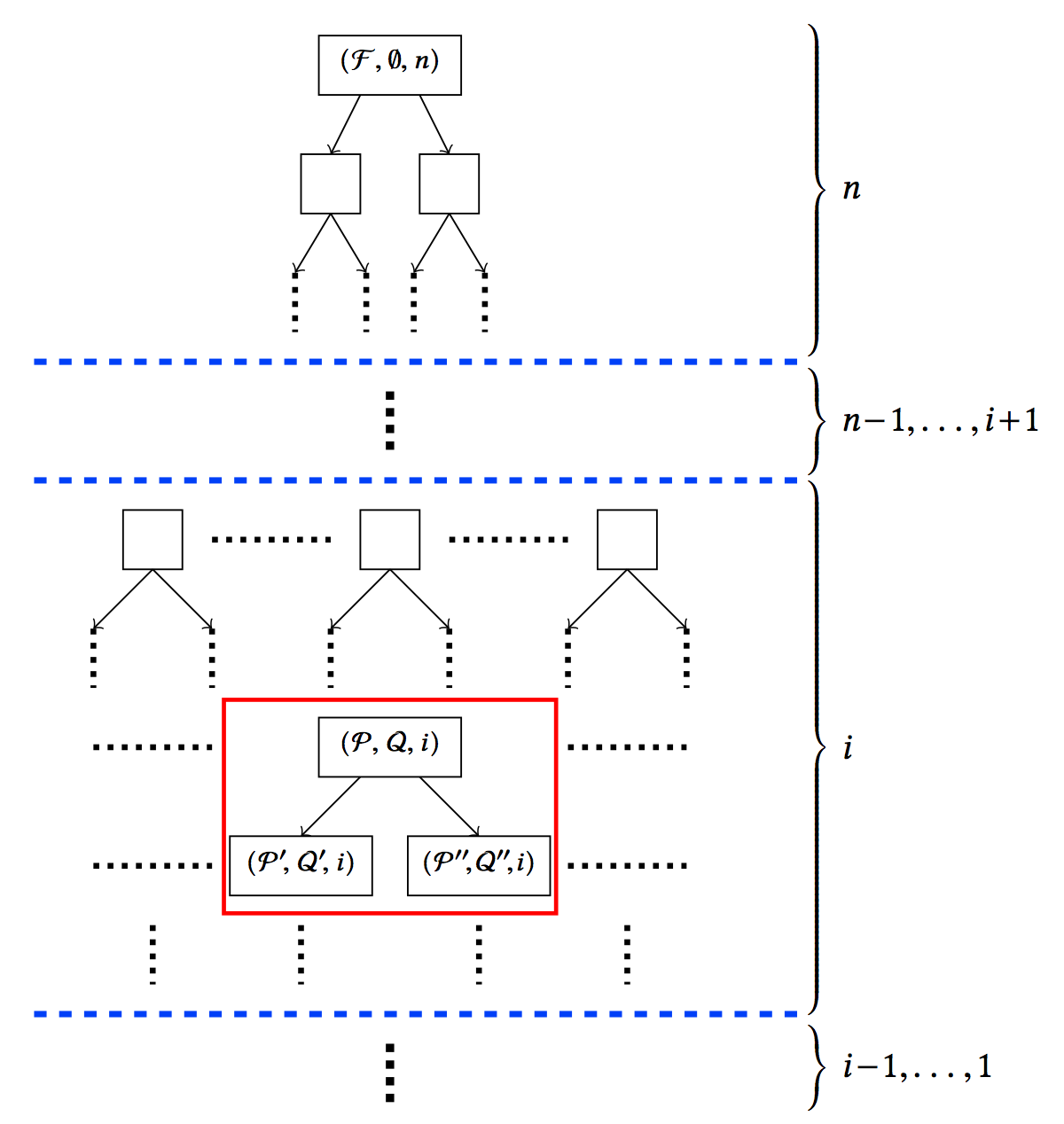}      
      \caption{Binary decomposition tree for Wang's method}
\label{fig:decTree}
    \end{figure}

    \begin{figure}[htp]
      \centering
\includegraphics[width=10cm,keepaspectratio]{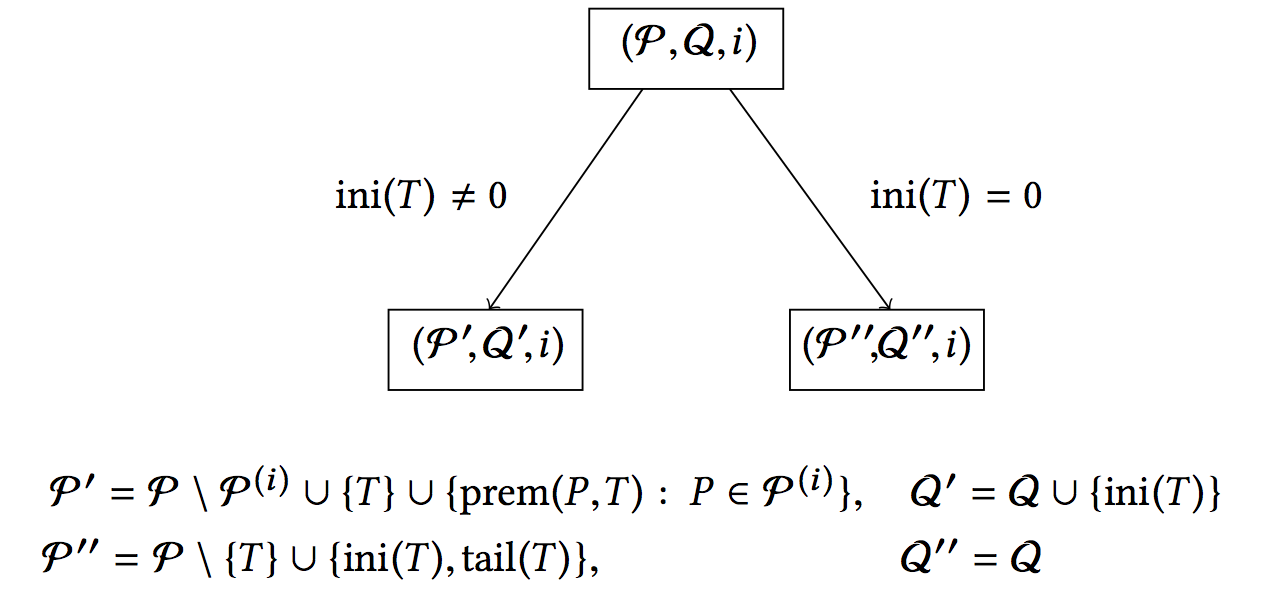}      
      \caption{Splitting at one node $(\pset{P}, \pset{Q}, i)$ in the binary decomposition tree}
\label{fig:oneNode}
    \end{figure}

\subsection{Chordality of polynomial sets in Wang's method}
\label{sec:graphs}

With a chordal polynomial set as input for Wang's method with respect to a perfect elimination ordering, the relationships are first clarified in the following propositions between the associated graphs of the polynomials in the left nodes and the associated graph of the input polynomial set and between the associated graphs of the polynomials in the right child nodes and those of the polynomials in the parent node. 

\begin{proposition}\label{prop:wang-left}
  Let $\pset{F} \subset \kx$ be a chordal polynomial set with $x_1 < \cdots < x_n$ as one perfect elimination ordering and $(\pset{P}, \pset{Q}, i)$ be any node in the binary decomposition tree for Wang's method applied to $\pset{F}$ such that $G(\pset{P}) \subset G(\pset{F})$, $T$ be a polynomial in $\pset{P}$ with minimal degree in $x_i$. Denote $\pset{P}' = \pset{P} \setminus \pset{P}^{(i)} \cup \{T\} \cup \{\prem(P, T):\ P\in \pset{P}^{(i)}\}$. Then $G(\pset{P}') \subset G(\pset{F})$.
\end{proposition}

\begin{proof}
  Clearly the vertexes of $G(\pset{P}')$ are also those of $G(\pset{P})$ and thus those of $G(\pset{F})$, and it suffices to prove that for any edge $(x_p, x_q) \in G(\pset{P}')$, we have $(x_p, x_q) \in G(\pset{F})$. 

Denote $\pset{R} := \{\prem(P, T):\, P\in \pset{P}^{(i)}\}$. Then by definition $\supp(\pset{R}^{(l)}) \subset \supp(\pset{P}^{(i)})$ for any $l=1, \ldots, i$. Furthermore, the following relationships hold: $\pset{P}'^{(l)} = \pset{P}^{(l)} \cup \pset{R}^{(l)}$ for $l=1, \ldots, i-1$ and $\pset{P}'^{(i)} = \{T\} \cup \pset{R}^{(i)}$. For any edge $(x_p, x_q) \in G(\pset{P}')$, there exist an integer $k~(p, q \leq k \leq i)$ and a polynomial $P \in \pset{P}'^{(k)}$ such that $x_p, x_q\in \supp(P)$. 

In the case when $k=i$, we have $(x_p, x_q)\in G(\pset{P}'^{(i)})$. Then $x_p, x_q \in \supp(T) \cup \supp(\pset{R}^{(i)}) \subset \supp(\pset{P}^{(i)})$ and thus $(x_p, x_i), (x_q, x_i) \in G(\pset{P}^{(i)}) \subset G(\pset{F})$. By the chordality of $\pset{F}$, we have $(x_p, x_q) \in G(\pset{F})$.

In the case when $k<i$, we have $(x_p, x_q) \in G(\pset{P}'^{(k)})$ with $\pset{P}'^{(k)} = \pset{P}^{(k)} \cup \pset{R}^{(k)}$. If $P \in \pset{P}^{(k)}$, then it is obvious that $(x_p, x_q) \in G(\pset{P}^{(k)}) \subset G(\pset{F})$; otherwise if $P\in \pset{R}^{(k)}$, then $x_p, x_q \in \supp(R^{(k)}) \subset \supp(\pset{P}^{i})$ and thus $(x_p, x_i), (x_q, x_i) \in G(\pset{P}) \subset G(\pset{F})$, then by the chordality of $\pset{F}$, we have $(x_p, x_q) \in G(\pset{F})$.
\end{proof}

\begin{proposition}\label{prop:wang-right}
Let $(\pset{P}, \pset{Q}, i)$ be any node in the binary decomposition tree for Wang's method and $T$ be a polynomial in $\pset{P}^{(i)}$ with minimal degree in $x_i$. Denote $\pset{P}'' = \pset{P} \setminus \{T\} \cup \{\ini(T), \tail(T)\}$. Then $G(\pset{P}'') \subset G(P)$. In particular, if $\supp(\tail(T)) = \supp(T)$, then $G(\pset{P}'') = G(\pset{P})$. 
\end{proposition}

\begin{proof}
  Since $\pset{P}''$ is constructed by replacing $T$ in $\pset{P}$ with $\ini(T)$ and $\tail(T)$, we only need to study the differences between $G(\pset{P})$ and $G(\pset{P}'')$ caused by this replacement. First, by $\supp(\ini(T)) \cup \supp(\tail(T)) \subset \supp(T)$ we have $\supp(\pset{P}'') \subset \supp(\pset{P})$. Second, for any edge $(x_p, x_q)$ in $G(\ini(T))$ or in $G(\tail(T))$, we have $(x_p, x_q) \in G(T)$, which means that all the edges of $G(\pset{P}'')$ are also edges of $G(\pset{P})$. Therefore, $G(\pset{P}'') \subset G(\pset{P})$.

In particular, if $\supp(\tail(T)) = \supp(T)$, then $\supp(\ini(T)) \cup \supp(\tail(T)) = \supp(T)$ and any edge $(x_p, x_q) \in \supp(T)$ is also contained in $G(\tail(T))$, and thus $G(\pset{P}'') = G(\pset{P})$. 
\end{proof}

\begin{example}\label{ex:wang-right}
  Let 
  \begin{equation*}
    \begin{split}
\pset{P}_1 &= [x_1+x_2, x_1+x_3, x_2+x_3, x_4^3+x_1, x_3x_4^2+x_3+x_4], \\
\pset{P}_2 &= [x_1+x_2, x_1+x_3, x_2+x_3, x_4^3+x_1, x_3x_4^2+x_4].
    \end{split}
  \end{equation*}
Then $G(\pset{P}_1) = G(\pset{P}_2)$ is shown in Figure~\ref{fig:tail} below (left). Let $\pset{P}_1''$ and $\pset{P}_2''$ be constructed from $\pset{P}_1$ and $\pset{P}_2$ with respect to $x_4$ respectively. Then $x_3x_4^2+x_3+x_4$ and $x_3x_4^2+x_4$ are chosen as $T$ respectively and
  \begin{equation*}
    \begin{split}
\pset{P}_1'' &= [x_1+x_2, x_1+x_3, x_2+x_3, x_3, x_4^3+x_1, x_3+x_4], \\
\pset{P}_2'' &= [x_1+x_2, x_1+x_3, x_2+x_3, x_3, x_4^3+x_1, x_4].
    \end{split}
  \end{equation*}
One may check that $G(\pset{P}_1'') = G(\pset{P}_1)$ while $G(\pset{P}_2'') \neq G(\pset{P}_2)$, with $G(\pset{P}_2'')$ shown in Figure~\ref{fig:tail} below (right).

    \begin{figure}[ht]
      \centering
\includegraphics[width=3cm,keepaspectratio]{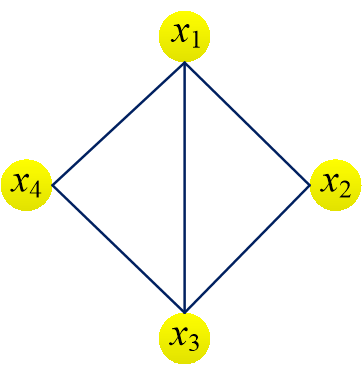}~\qquad 
\includegraphics[width=3cm,keepaspectratio]{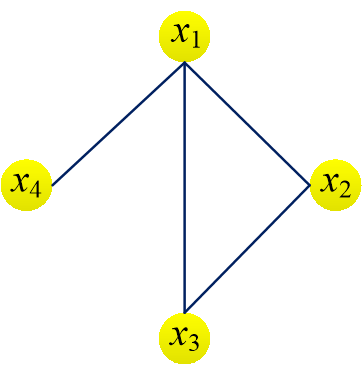}      
      \caption{The associated graphs $G(\pset{P}_1) = G(\pset{P}_2) = G(\pset{P}_1'')$ (left) and $G(\pset{P}_2'')$ (right) in Example~\ref{ex:wang-right}}
      \label{fig:tail}
\end{figure}
\end{example}

Next we prove that with a chordal input polynomial set, all the polynomials in the nodes of the decomposition tree of Wang's method, and thus all the computed triangular sets, have associated graphs which are subgraphs of that of the input polynomial set.

  \begin{theorem}\label{thm:wang}
Let $\pset{F} \subset \kx$ be a chordal polynomial set with $x_1 < \cdots < x_n$ as one perfect elimination ordering and $(\pset{P}, \pset{Q}, i)$ be any node in the binary decomposition tree for Wang's method applied to $\pset{F}$. Then $G(\pset{P}) \subset G(\pset{F})$.
  \end{theorem}

  \begin{proof}
     We induce on the depth $d$ of $(\pset{P}, \pset{Q}, i)$ in the binary decomposition tree. When $d=1$, then $(\pset{P}, \pset{Q}, i)$ is a child node of $(\pset{F}, \emptyset, n)$, and $G(\pset{P}) \subset G(\pset{F})$ by Proposition~\ref{prop:wang-left} if it is a left child node or by Proposition~\ref{prop:wang-right} otherwise. 
     Now assume that the first polynomial in any node of depth $d$ in the decomposition tree has an associated graph which is a subgraph of $G(\pset{F})$. Let $(\pset{P}, \pset{Q}, i)$ be of depth $d+1$ and $(\tilde{\pset{P}}, \tilde{\pset{Q}}, i)$ be its parent of depth $d$ in the decomposition. Then $G(\pset{P}) \subset G(\pset{F})$ by Proposition~\ref{prop:wang-left} if $(\pset{P}, \pset{Q}, i)$ is a left child node or $G(\pset{P}) \subset G(\tilde{\pset{P}}) \subset G(\pset{F})$ by Proposition~\ref{prop:wang-right} otherwise. This ends the inductive proof. 
  \end{proof}

  \begin{corollary}\label{cor:wang-ts}
    Let $\pset{F} \subset \kx$ be a chordal polynomial set with $x_1 < \cdots < x_n$ as one perfect elimination ordering and $\pset{T}_1, \ldots, \pset{T}_r$ be the triangular sets computed by Wang's method applied to $\pset{F}$. Then $G(\pset{T}_i) \subset G(\pset{F})$ for $i=1, \ldots, r$. 
  \end{corollary}

  \begin{proof}
    Straightforward from Theorem~\ref{thm:wang} with the fact that each triangular set $\pset{T}_j$ for some $i~(1\leq i \leq r)$ is from a node $(\pset{T}_i, \pset{Q}_i, 0)$ in the decomposition tree such that $\pset{T}_i$ contains no non-zero constant. 
  \end{proof}

\subsection{An illustrative example}
\label{sec:wangExample}

Here we illustrate the changes of chordality of polynomials computed in the triangular decomposition via Wang's method applied to
\begin{equation}
  \label{eq:illus}
\pset{F} = \{x_2+x_1+2, (x_2+2)x_3+x_1, (x_3+x_2)x_4 + x_3 -1, x_4+x_2\}
\end{equation}
in $\qnum[x_1, x_2, x_3, x_4]$ for the variable ordering $x_1 < x_2 < x_3 < x_4$. The associated graph $G(\pset{F})$ is shown in Figure~\ref{fig:exa1}, and one can check that $G(\pset{F})$ is chordal with $x_1 < x_2 < x_3 < x_4$ as one perfect elimination ordering. 

\begin{figure}[ht]
      \centering
\includegraphics[width=3cm,keepaspectratio]{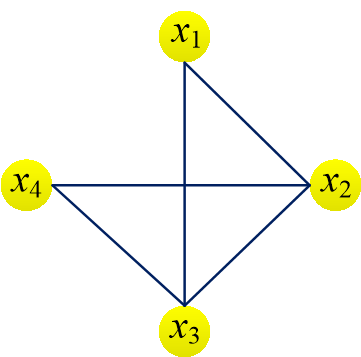}    
      \caption{The associated graph $G(\pset{F})$ with $\pset{F}$ from \eqref{eq:illus}}
      \label{fig:exa1}
    \end{figure}

First $T = (x_3+x_2)x_4+x_3-1$ is chosen as the polynomial in $\pset{F}^{(4)}$ with minimal degree in $x_4$, then a new polynomial set $\pset{F}' = \{x_2+x_1+2, (x_2+2)x_3+x_1, (x_3+x_2), x_3-1, x_4+x_2\}$, which corresponds to the right child node for the case $\ini(T) = x_3+x_2 = 0$ in the binary decomposition tree, is added to $\Phi$ for further computation. Psuedo division of $x_4+x_2$ over $T$ with respect to $x_4$ is performed to result in 
$$\pset{P} = \{x_2+x_1+2, (x_2+2)x_3+x_1, (x_2-1)x_3 + x_2^2 +1, (x_3+x_2)x_4 + x_3 -1\},$$
and thus the left child node is $(\pset{P}, \{x_3+x_2\}, 3)$ in the binary tree.

Next $T' = ((x_2+2)x_3+x_1)$ is chosen as the polynomial in $\pset{P}^{(3)}$ with minimal degree in $x_3$, then a new polynomial set $\pset{F}'' = \{x_1, x_2+x_1+2, x_2+2, (x_2-1)x_3 + x_2^2 +1, (x_3+x_2)x_4 + x_3 -1\}$ is added to $\Phi$, and the pseudo division of $(x_2-1)x_3 + x_2^2 +1$ over $(x_2+2)x_3+x_1$ results in 
$$\pset{P}' = \{x_2+x_1+2, x_2^3 + 2x_2^2 - (x_1-1)x_2 + x_1+2, (x_2+2)x_3+x_1, (x_3+x_2)x_4 + x_3 -1\}$$
and the left node is $(\pset{P}', \{x_x+x_2, x_2+2\}, 2)$.

At this step $T'' = x_2+x_1+2$ is chosen as the polynomial in $\pset{P}'^{(2)}$ with minimal degree in $x_2$, then no polynomial set is added to $\Phi$ since $\ini(T'') = 1$ and the pseudo-division of $x_2^3 + 2x_2^2 - (x_1-1)x_2 + x_1+2$ over $T''$ with respect to $x_2$ results in the first triangular set 
\begin{equation}
  \label{eq:illus-T1}
\pset{T}_1 = [-x_1^3 + x_1^2 + 14x_1 + 16, x_2+x_1+2, (x_2+2)x_3+x_1, (x_3+x_2)x_4 + x_3 -1].
\end{equation}
With similar treatments on $\pset{F}'$ and $\pset{F}''$ in $\Phi$, the other two triangular sets 
\begin{equation}
  \label{eq:illus-T23}
  \begin{split}
    \pset{T}_2 &= [x_1+1, x_2+1, x_3-1, x_4+x_2], \\
    \pset{T}_3 &= [x_1, x_2+2, (x_2-1)x_3+x_2^2+1, (x_3+x_2)x_4+x_3-1]
  \end{split}
\end{equation}
are computed. 

The associated graphs of all these three computed triangular sets are shown in Figure~\ref{fig:exa2}. One can find that the associated graphs $G(\pset{F})$ and $G(\pset{T}_1)$ are the same, while $G(\pset{T}_2)$ and $G(\pset{T}_3)$ are strict subgraphs of $G(\pset{F})$.

\begin{figure}[ht]
      \centering
\includegraphics[width=3cm,keepaspectratio]{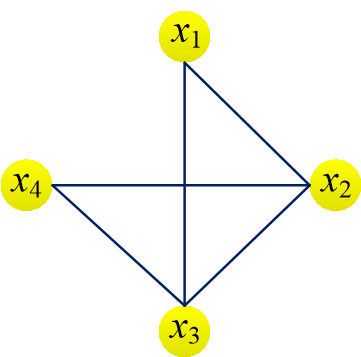}~\qquad
\includegraphics[width=3cm,keepaspectratio]{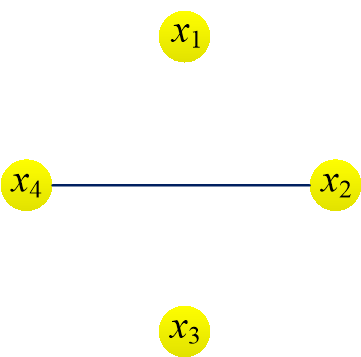}~\qquad
\includegraphics[width=3cm,keepaspectratio]{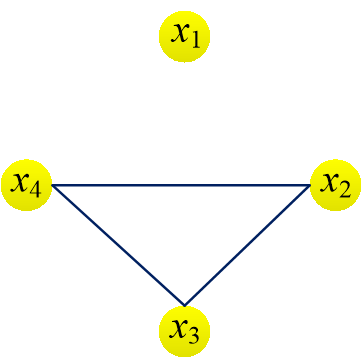}
      \caption{The associated graphs $G(\pset{T}_1)$ (left), $G(\pset{T}_2)$ (middle), and $G(\pset{T}_3)$ (right) with $\pset{T}_1$ from \eqref{eq:illus-T1} and $\pset{T}_2, \pset{T}_3$ from \eqref{eq:illus-T23}}
      \label{fig:exa2}
    \end{figure}

\section{Further remarks on the applications of graphs structures of polynomial sets}
\label{sec:discussion}

\subsection{Variable sparsity of polynomial sets}
\label{sec:sparsity}

When referring to a polynomial set $\pset{F} \subset \kx$ to be sparse, one usually means that the percentage of terms effectively appearing in $\pset{F}$ in all the possible terms in the variables $x_1, \ldots, x_n$ up to a certain degree is low. This kind of sparsity for polynomial sets is convenient for the computation of \grobner bases which is essentially based on reduction with respect to terms. In fact, efficient algorithms for computing \grobner bases for sparse polynomial sets defined in this way have been proposed, implemented, and analyzed \cite{FSS2014s}. 

Instead of terms of polynomials, triangular sets focus on the variables of polynomials. As exploited in \cite{C2017c}, sparsity of the polynomial sets with respect to their variables are partially reflected in their associated graphs. To make it precise, let $G(\pset{F}) = (V, E)$ be the associated graph of a polynomial set $\pset{F} = \{F_1, \ldots, F_r\} \subset \kx$. Then the \emph{variable sparsity} $s_v(\pset{F})$ of $\pset{F}$ can be defined as 
$$s_v(\pset{F}) = |E| / \binom{2}{|V|},$$
where the denominator is the number of edges of a complete graph composed of $|V|$ vertexes. 

Furthermore, the associated graph $G(\pset{F})$ can be extended to a weighted one $G^w(\pset{F})$ by associating the number $\#\{F\in \pset{F}: x_i, x_j \in \supp(F)\}$ to each edge $(x_i, x_j)$ of $G(\pset{F})$. Let $G^{w}(\pset{F}) = (V, E)$, with the weight $w_e$ for each $e\in E$, be the weighted associated graph of $\pset{F}$. Then the \emph{weighted variable sparsity} $s_v^w(\pset{F})$ of $\pset{F}$ can be defined as
$$ s_v^w(\pset{F}) = \frac{ \sum_{e\in E} w_e}{r\cdot \binom{2}{|V|}},$$
where $r$ is the number of polynomials in $\pset{F}$.

\subsection{Complexity analysis for triangular decomposition in top-down style}
\label{sec:complexity}
In general, due to the complicated behaviors in the decomposition process, the complexity of triangular decomposition is not as clearly known as that of computation of \grobner bases \cite{A1999c,P2013o,HSL2014o,BFS04,BFSS2013o,FM17}. 

For a graph $G$, another graph $G'$ is called a \emph{chordal completion} if $G'$ is chordal with $G$ as its subgraph. The \emph{treewidth} of a graph $G$ is defined to be the minimum of the sizes of the largest cliques in all the possible chordal completions of $G$. It has been shown that many NP-complete problems related to graphs can be solved efficiently if the graphs have bounded treewidth \cite{AP1989l}. 

As shown by Theorem~\ref{thm:wang}, when the input polynomial set is chordal, the associated graphs of all the polynomials in the decomposition process of Wang's method are subgraphs of the chordal associated graph of the input polynomial set. In other words, the input chordal associated graph imposes some kind of upper bound for all the polynomials in the decomposition process. Furthermore, the complexity of computing \grobner bases has been analyzed for polynomial systems by using the treewidth of their associated graphs \cite{C2016e}. These comments lead to the hope of refined complexity analysis of triangular decomposition in top-down style, especially on Wang's method, from the viewpoint of chordal graphs and their treewidth.

\smallskip
\noindent {\bf Acknowledgements} The authors would like to thank Dongming Wang and Diego Cifuentes for helpful discussions on Wang's methods for triangular decomposition and on chordal graph structures of polynomial sets in triangular decomposition respectively.

% \bibliographystyle{plain}
% \bibliography{chordality} 

\end{document}